\documentclass[10pt]{amsart}
\usepackage[T2A]{fontenc}
\usepackage[utf8]{inputenc}
\usepackage[english]{babel}
\usepackage{amsmath}
\usepackage{amssymb}
\usepackage{amsfonts}
\usepackage{mathrsfs}
\usepackage{graphicx}
\usepackage{color}
\usepackage{commath}
\usepackage{hyperref}
\usepackage{caption}

\DeclareMathOperator{\CC}{\mathbb{C}}

\DeclareMathOperator{\RR}{\mathbb{R}}

\DeclareMathOperator{\QQ}{\mathbb{Q}}

\DeclareMathOperator{\ZZ}{\mathbb{Z}}

\DeclareMathOperator{\GC}{\mathcal{G}}

\DeclareMathOperator{\BC}{\mathcal{B}}
\DeclareMathOperator{\NC}{\mathcal{N}}
\DeclareMathOperator{\VC}{\mathcal{V}}
\DeclareMathOperator{\WC}{\mathcal{W}}
\DeclareMathOperator{\TC}{\mathcal{T}}
\DeclareMathOperator{\EC}{\mathcal{E}}
\DeclareMathOperator{\PC}{\mathcal{P}}

\DeclareMathOperator{\LC}{\mathcal{L}}
\DeclareMathOperator{\HC}{\mathcal{H}}
\DeclareMathOperator{\MC}{\mathcal{M}}
\DeclareMathOperator{\AC}{\mathcal{A}}
\DeclareMathOperator{\FC}{\mathcal{F}}

\DeclareMathOperator{\JC}{\mathcal{J}}
\DeclareMathOperator{\IC}{\mathcal{I}}
\DeclareMathOperator{\CCal}{\mathcal{C}}

\DeclareMathOperator{\RS}{\mathscr{R}}

\DeclareMathOperator{\PS}{\mathscr{P}}

\DeclareMathOperator{\fG}{\mathfrak{f}}

\DeclareMathOperator{\rank}{rank}

\DeclareMathOperator{\vertex}{vert}

\DeclareMathOperator{\size}{size}
\DeclareMathOperator{\poly}{poly}

\DeclareMathOperator{\tcone}{tcone}

\DeclareMathOperator{\lmod}{\text{\it\quad modulo polyhedra with lines}}

\DeclareMathOperator{\toddp}{td}

\DeclareMathOperator{\proj}{proj}

\DeclareMathOperator{\BUnit}{\mathbf 1}
\DeclareMathOperator{\BZero}{\mathbf 0}
\DeclareMathOperator{\xB}{\mathbf{x}}

\DeclareMathOperator{\SNF}{\text{\rm SNF}}

\newcommand*{\intint}[2][1]{\{#1, \dots, #2\}}

\newcommand\restr[2]{{
  \left.\kern-\nulldelimiterspace 
  #1 
  \vphantom{\big|} 
  \right|_{#2} 
  }}
\newtheorem{lemma}{Lemma}
\newtheorem{theorem}{Theorem}
\newtheorem{definition}{Definition}
\newtheorem{corollary}{Corollary}

\newtheorem{problem}{Problem}
\newtheorem{remark}{Remark}

\begin{document}

\title[Faster counting in $\Delta$-modular polyhedra (corrected version)]{A faster algorithm for counting the integer points number in $\Delta$-modular polyhedra (corrected version)}
\author{{D.~V.~Gribanov}}%
\address{Dmitry Gribanov
\newline\hphantom{iii} National Research University Higher School of Economics,
\newline\hphantom{iii} 25/12 Bolshaja Pecherskaja Ulitsa,
\newline\hphantom{iii} 603155, Nizhny Novgorod, Russia}%
\email{dimitry.gribanov@gmail.com}%
\author{{I.~A.~Shumilov}}%
\address{Ivan Shumilov
\newline\hphantom{iii} Lobachevsky State University of Nizhny Novgorod,
\newline\hphantom{iii} 23 Gagarina ave., Nizhny Novgorod, 603950, Russia}
\email{ivan.a.shumilov@gmail.com}%
\author{{D.~S.~Malyshev}}%
\address{Dmitry Malyshev
\newline\hphantom{iii} National Research University Higher School of Economics,
\newline\hphantom{iii} 25/12 Bolshaja Pecherskaja Ulitsa,
\newline\hphantom{iii} 603155, Nizhny Novgorod, Russia;
\newline\hphantom{iii} Lobachevsky State University of Nizhny Novgorod,
\newline\hphantom{iii} 23 Gagarina ave., Nizhny Novgorod, 603950, Russia}
\email{dsmalyshev@rambler.ru}%

\thanks{The article was prepared under financial support of Russian Science Foundation grant No 21-11-00194.}

\maketitle {\small
\begin{quote}
\noindent{\sc Abstract. } Let a polytope $\PC$ be defined by a system $A x \leq b$. We consider the problem of counting the number of integer points inside $\PC$, assuming that $\PC$ is $\Delta$-modular, where the polytope $\PC$ is called $\Delta$-modular if all the rank sub-determinants of $A$ are bounded by $\Delta$ in the absolute value. We present a new FPT-algorithm, parameterized by $\Delta$ and by the maximal number of vertices in $\PC$, where the maximum is taken by all r.h.s. vectors $b$. We show that our algorithm is more efficient for $\Delta$-modular problems than the approach of A.~Barvinok et al. \cite{BarvBook,BarvPom,BarvWoods,OnBarvinoksAlg_Dyer,HalfOpen}. To this end, we do not directly compute the short rational generating function for $\PC \cap \ZZ^n$, which is commonly used for the considered problem. Instead, we use the dynamic programming principle to compute its particular representation in the form of exponential series that depends on a single variable. We completely do not rely to the Barvinok's unimodular sign decomposition technique. 

Using our new complexity bound, we consider different special cases that may be of independent interest. For example, we give FPT-algorithms for counting the integer points number in $\Delta$-modular simplices and similar polytopes that have $n + O(1)$ facets. As a special case, for any fixed $m$, we give an FPT-algorithm to count solutions of the unbounded $m$-dimensional $\Delta$-modular subset-sum problem. 

 \end{quote}
}

\section{Introduction}

The paper is structured as follows. In Subsection \ref{results_brief}, we briefly discuss the most important consequences of our work. In Subsection \ref{basic_defs_subs}, we give basic definitions and notations that are necessary for further understanding. In Subsection \ref{problems_subs}, we define the main problem and provide a detailed description of our results. In Subsection \ref{results_comp_subs}, we compare our results with other works in the field. In Subsection \ref{algebra_subs}, we give some basics of the polyhedral algebra that will be used further. Finally, Sections \ref{all_non_zero_proof} and \ref{main_th1_sec} are devoted to the proofs of our results.

\subsection{Brief discussion of our results }\label{results_brief}

Let a polytope $\PC$ be defined by one of the following ways: 
\begin{enumerate}
\item[(i)] {\bf System in the canonical form:} $\PC = \{x \in \RR^n \colon A x \leq b\}$, where $A \in \ZZ^{(n+m) \times n}$, $b \in \QQ^{(n+m)}$, $\rank(A) = n$ and $d :=\dim(\PC) = n$;
\item[(ii)] {\bf System in the standard form:} $\PC = \{x \in \RR_+^n \colon A x = b\}$, where $A \in \ZZ^{m \times n}$, $b \in \ZZ^{m}$, $\rank(A) = m$ and $d := \dim(\PC) = n-m$;
\end{enumerate} 
and let all the rank-order sub-determinants of $A$ be bounded by $\Delta$ in absolute values. We show that $|\PC \cap \ZZ^n|$ can be computed with an algorithm, having the arithmetic complexity bound
$$
O\bigl( \nu^2 \cdot d^4 \cdot \Delta^4 \cdot \log(\Delta) \bigr),
$$ where $\nu = \nu(A)$ is the maximal possible number of vertices in a $d$-dimensional polytope $\PC$ defined by one of the systems above with a fixed matrix $A$ and a varying r.h.s. vector $b$. For small values of $\Delta$ and growing $d$ our new complexity bound can be more effective than the state of the art bound, due to Barvinok et al. \cite{BarvBook,BarvPom,BarvWoods,OnBarvinoksAlg_Dyer,HalfOpen}:
$$
\nu \cdot 2^{O(d)} \cdot \bigl(\log_2 (\Delta) \bigr)^{d \ln (d)}.
$$
Using different ways to estimate the parameter $\nu$, we can present different complexity bounds: 
\begin{itemize}
    \item A bound that is polynomial on $n$ and $\Delta$, for any fixed $m$. Taking $m=1$, it gives $O\bigl( d^6 \cdot \Delta^4 \cdot \log(\Delta) \bigr)$ arithmetic complexity FPT-algoritm to compute the number of integer points in a simplex or the number of solutions of the unbounded subset-sum problem, where $\Delta$ means the maximal weight of an item;
    
    \item A bound that is polynomial on $m$ and $\Delta$, for any fixed $n$. This bound can be used to obtain an even faster algorithm for the ILP feasibility problem, when the parameters $m$ and $\Delta$ are relatively small. For example, if $m = O(d)$ and $\Delta = 2^{O(d)}$, our bound becomes $2^{O(d)}$, which is faster than the recent state of the art algorithm due to Ries \& Rothvoss \cite{log_ILP} with the complexity bound $\log(d)^{O(d)} \cdot \poly\bigl(\size(A,b)\bigr)$;
    
    \item For $\Delta = d^{O(1)}$, our bound becomes $d^{O(d)}$.
\end{itemize}

\subsection{Basic definitions and notations}\label{basic_defs_subs}

We denote by $A_{ij}$ its $ij$-th element, by $A_{i*}$ its $i$-th row, and by $A_{*j}$ its $j$-th column. For subsets $I$ and $J$, the symbol $A_{I J}$ denotes the sub-matrix of $A$, which is generated by all the rows with indices in $I$ and all the columns with indices in $J$. If $I$ or $J$ is replaced by $*$, then all the rows or columns are selected, respectively. Sometimes, we simply write $A_{I}$ instead of $A_{I*}$ and $A_{J}$ instead of $A_{*J}$, if this does not lead to confusion. The maximum absolute value of entries of a matrix $A$ is denoted by $\|A\|_{\max} = \max\limits_{i,j} \abs{A_{i\,j}}$. The $l_p$-norm of a vector $x$ is denoted by $\norm{x}_p$. The number of non-zero components of a vector $x$ is denoted by $\norm{x}_0=\abs{\{i\colon x_i \not= 0\}}$. 
For $x \in \RR$, we denote by $\lfloor x \rfloor$, $\{x\}$, and $\lceil x \rceil$ the floor, fractional part, and ceiling of $x$, respectively. 
For $c,x \in \RR^n$, by $\langle c, x \rangle$ we denote the standard scalar product of two vectors. In other words, $\langle c, x \rangle = c^\top x$.

Let $S \in \ZZ_{\geq 0}^{n \times n}$ be a diagonal matrix and $v \in \ZZ^n$. We denote by $v \bmod S$ the vector, whose $i$-th component equals $v_i \bmod S_{i i}$. For $\MC \subseteq \ZZ^n$, we denote $\MC \bmod\, S = \{ v \bmod S \colon v \in \MC \}$. For example, the set $\ZZ^n \bmod\, S$ consists of $\det(S)$ elements.

\begin{definition}
For a matrix $A \in \ZZ^{m \times n}$, by $$
\Delta_k(A) = \max\left\{\abs{\det (A_{IJ})} \colon I \subseteq \intint m,\, J \subseteq \intint n,\, |I| = |J| = k\right\},
$$ we denote the maximum absolute value of determinants of all the $k \times k$ sub-matrices of $A$.  By $\Delta_{\gcd}(A,k)$ we denote the greatest common divisor of determinants of all the $k \times k$ sub-matrices of $A$. Additionally, let $\Delta(A) = \Delta_{\rank(A)}(A)$ and $\Delta_{\gcd}(A) = \Delta_{\gcd}(A,\rank(A))$.
If $\Delta(A) \leq \Delta$, for some $\Delta > 0$, then $A$ is called \emph{$\Delta$-modular}.
\end{definition}

\begin{definition}
For arbitrary $A \in \ZZ^{m \times n}$ and $b \in \QQ^m$, we denote
\begin{gather*}
    \PC(A,b) = \{x \in \RR^n \colon A x \leq b\},\\
    \nu(A) = \max\limits_{b \in \QQ^m} \bigl\{\vertex\bigl(\PC(A,b)\bigr)\bigr\},\\
    \xB^z = x_1^{z_1} \cdot \ldots \cdot x_n^{z_n}, \quad\text{and}\quad  \fG(\PC;\xB) = \sum\limits_{z \in \PC \cap \ZZ^n} \xB^z.
\end{gather*}
\end{definition}


\subsection{Formal definition the lattice point counting problem and the detailed description of our results}\label{problems_subs}


In this paper, we consider the problem to count integer points in a polyhedron, which is defined as follows: 
\begin{problem}\label{main_prob}
Compute the value of $|\PC \cap \ZZ^n|$ for a rational polytope $\PC$ defined by one of the following ways:
\begin{enumerate}
    \item The polytope $\PC$ is defined by \emph{a system in the canonical form}: $\PC = \{x \in \RR^n \colon A x \leq b\}$, where $A \in \ZZ^{(n+m)\times n}$, $b \in \QQ^{n+m}$, and $\dim(\PC) = \rank(A) = n$;
    
    \item The polytope $\PC$ is defined by \emph{a system in the standard form}: $\PC = \{x \in \RR_{\geq 0}^n\colon A x = b\}$, where $A \in \ZZ^{m \times n}$, $b \in \QQ^{m}$, $\rank(A) = m$, $\dim(\PC) = n - m$ and $\Delta_{\gcd}(A) = 1$.
\end{enumerate}
\end{problem}
Our main results are presented in the following theorem and its corollary:
\begin{theorem}\label{main_th1}
Problem \ref{main_prob} can be solved with a randomized algorithm, having the expected arithmetic complexity bound 
$$
O\bigl( \nu^2 \cdot d^4 \cdot \Delta^4 \cdot \log(\Delta) \bigr),
$$ where $\Delta = \Delta(A)$, $d = \dim(\PC)$ ($d = n$, for the canonical form, and $d = n - m$, for the standard form) and $\nu = \nu(A)$.
\end{theorem}
The value of $\nu$ can be estimated in terms of $m$ due to the seminal work \cite{MaxFacesTh} of McMullen. In turn, due to Lee, Paat, Stallknecht \& Luze \cite{ModularDiffColumns}, the value of $m$ can be estimated in terms of $n$ and $\Delta$. Using these observations, we propose new complexity bounds for Problem \ref{main_prob} in the following Corollary \ref{main_corr1}. Additionally, we show how to handle the case of unbounded polyhedra.
\begin{corollary}\label{main_corr1}
The following complexity bounds hold for the Problem \ref{main_prob}:
\begin{enumerate}
    \item The bound $ O\bigl( \frac{d}{m} +1 \bigr)^{2m} \cdot d^4 \cdot \Delta^4 \cdot \log(\Delta)$ that is polynomial in $d$ and $\Delta$, for any fixed $m$; 
    
    \item The bound $
    O\bigl( \frac{m}{d} + 1 \bigr)^{d} \cdot d^4 \cdot \Delta^4 \cdot \log(\Delta)
    $ that is polynomial in $m$ and $\Delta$, for any fixed $d$;
    
    \item The bound $
    O(d)^{4 + d} \cdot \Delta^{4+2d} \cdot \log(\Delta)
    $ that is polynomial in $\Delta$, for any fixed $d$.
\end{enumerate}
To handle the case, when $\PC$ is an unbounded polyhedron, we need to pay an additional factor of $O(\frac{d}{m} +1)^2 \cdot d^4$ in the first bound and $O(d^4)$ in the second bound. The third bound stays unchanged.
\end{corollary}
Proofs of Theorem \ref{main_th1} and Corollary \ref{main_corr1} will be given in Section \ref{main_th1_sec} and Subsection \ref{main_corr1_subs}, respectively. Taking $m = 1$, the first bound can be used to count the number of integer points in a simplex or the number of solutions of \emph{the unbounded subset-sum problem} $w^\top x = w_0,\, x \in \ZZ^n_{\geq 0}$. For both problems, it gives the arithmetic complexity bound $
O\bigl(n^6 \cdot \Delta^4 \cdot \log(\Delta)\bigr), 
$ where $n$ is a number of items and $\Delta = \|w\|_{\infty}$ with respect to the subset-sum problem.


The second bound can be used to obtain a faster algorithm for the ILP feasibility problem, when the parameters $m$ and $\Delta$ are relatively small. For example, taking $m = O(d)$ and $\Delta = 2^{O(d)}$ in the second bound, it becomes $2^{O(d)}$, which is faster than the very recent $\log(d)^{O(d)} \cdot \poly\bigl(\size(A,b)\bigr)$-complexity algorithm from the breakthrough work \cite{log_ILP} due to Reis \& Rothvoss. The previous state of the art algorithm is due to Dadush, Peikert \& Vempala \cite{DadushDis,DadushFDim} (see also \cite{Convic,ConvicComp,DConvic}, for a bit more general setting), and it has the complexity bound $O(d)^d \cdot \poly(\size(A,b))$. Good surveys on the related $\Delta$-modular ILP problems and parameterised ILP complexity are given in \cite{AlgorithmicTheoryILP,OnCanonicalProblems_Grib,CountingFixedM,FiveMiniatures}. 



    
    


\begin{remark}\label{forms_reduction_rm}
We are interested in development of algorithms that will be polynomial, when we bound some of the parameters $d$, $m$, and $\Delta$.
Due to \cite[Corollary~3]{CountingFixedM}, the problem in the standard form can be polynomially reduced to the problem in the canonical form maintaining values of $m$ and $\Delta$, see also \cite[Lemmas~4 and~5]{OnCanonicalProblems_Grib} and \cite{BoroshTreybigProof} for a more general reduction. Hence, in the proofs we will only consider polytopes defined by systems in the canonical form. 
\end{remark}

\begin{remark}\label{gcd_standard_form_rm}
To simplify analysis, we assume that $\Delta_{\gcd}(A) = 1$ for ILP problems in the standard form. It can be done without loss of generality, because the original system $A x = b$,~$x \geq \BZero$ can be polynomially transformed to the equivalent system with $\Delta_{\gcd}(A) = 1$. For the justification see \cite[Remark~3]{CountingFixedM}.

\end{remark}

\begin{remark}\label{complexity_rm}
In our paper, we prove only arithmetical complexity bounds for presented algorithms. By the \emph{arithmetic complexity}, we mean the number of elementary arithmetic operations with integer values, whose size polynomially depends on  input size. Consequently, if some algorithm has a polynomial arithmetic complexity bound, then it has a polynomial bit-complexity bound. In our paper, when we say ``complexity'' (without ``arithmetic''), we mean bit-complexity bounds. 
\end{remark}

\subsection{Comparison with related works} \label{results_comp_subs}

The first polynomial-time complexity algorithm in a fixed dimension for problem \ref{main_prob} is due to Barvinok \cite{Barv_Original}. Further modifications were given in \cite{BarvPom,OnBarvinoksAlg_Dyer,HalfOpen}. A complete exposition of the Barvinok's approach could be found in \cite{BarvBook,BarvPom,continuous_discretely,AlgebracILP}, additional discussions and connections with "dual"-type counting algorithms could be found in the book \cite{counting_Lasserre_book} due to Lasserre.

Since any polytope can be transformed to an integer-equivalent simple polytope, using a slight perturbation of the r.h.s.\! vector $b$, we can estimate the complexity of algorithms based on the M.~Brions's formula using the parameter $\nu$ (for an algorithmic version of this statement, see Lemma \ref{poly_simplification_th} of our work). Alternatively, we can define $\nu$ as the maximum number of simple cones in the normal fun triangulation of $\PC$ for a fixed matrix $A$ and a varying r.h.s.\! vector $b$. 
Combining a single exponential algorithm for the shortest lattice vector problem from Micciancio \& Voulgaris \cite{SVP_exp} with analysis from \cite[Chapter~16]{BarvBook}, the arithmetical complexity of the Barvinok's algorithm can be estimated by
\begin{equation}\label{barvinok_complexity}
    \nu \cdot 2^{O(d)} \cdot \bigl(\log_2(\Delta)\bigr)^{d \ln(d)}.
\end{equation}
The paper \cite{simple_formula_counting} due to Lasserre \& Zeron gives formulae, based on the Gomory's group-theoretic approach, with the following complexity 
\begin{equation}\label{lasserre_complexity2}
\nu \cdot d^{O(1)} \cdot \Delta^{d}.    
\end{equation}
For $\Delta = O(d)$, the last complexity bound is better than \eqref{barvinok_complexity}. The goal of the paper \cite{CountingFixedM} is to develop an algorithm, whose complexity will be polynomial on $\nu$ and $d$, for any fixed $\Delta$. Its resulting arithmetic complexity can be estimated as
\begin{equation}\label{gribanov_complexity1}
    O\bigl( \nu \cdot d^2 \cdot d^{\log_2(\Delta)} \bigr).
\end{equation}
Our paper presents a new counting algorithm, whose arithmetic complexity to process a vertex is polynomial on $\nu$, $d$ and $\Delta$:
$$
    O\bigl(\nu^2 \cdot d^4 \cdot \Delta^4 \cdot \log(\Delta) \bigr).
$$ 
The big difference of our algorithm with the previous papers is that we do not compute the rational generating function $\fG(\PC; \xB)$ of the set $\PC \cap \ZZ^d$. Instead of that we directly compute an exponential generating function $\hat \fG(\PC; \tau)$ that depends only on one variable $\tau$. This exponential generating function can be obtained from the original rational generating function substituting $x_i = e^{c_i \tau}$, for $c \in \RR^d$. The new function forgets the structure of the set $\PC \cap \ZZ^d$, but it still can be used for counting. For example, two monomials $x_1^1 x_2^2$ and $x_1^2 x_2^1$ glue to one exponential term $2 e^{3 \tau}$ after the map $x_i = e^{c_i \tau}$ with $c = (1, 1)^\top$. At the current moment, it is a very interesting open question to compute the full rational generating function $\fG(\PC; \xB)$ of the set $\PC \cap \ZZ^d$ by an algorithm, whose complexity will polynomially depend on $\nu$, $d$, and $\Delta$.

Let us additionally discus other part of our method that can be of independent interest. For a given set $\AC$ of $m$ non-zero vectors in $\QQ^n$, let us consider the problem to compute a vector $z \in \ZZ^n$, such that $a^\top z \not= 0$, for all $a \in \AC$. Preferably, the value of $\|z\|_{\infty}$ should be as small as possible. Due to the original work of A.~Barvinok \cite{Barv_Original}, the vector $z$ could be found by a polynomial-time algorithm as a point on the moment curve. The paper \cite{HalfOpen} of K{\"o}ppe \& Verdoolaege gives an alternative method, based on "irrational decompositions", from the work \cite{IrrationalDecomp} of K{\"o}ppe. The considered polynomial-time methods can generate vectors $z$ with the only guaranty $\|z\|_{\infty} \leq M^n$, for some constant $M \geq m$. However, due to De~Loera, Hemmecke, Tauzer \& Yoshida \cite{EffectiveCounting}, the vector $z$ with sufficiently small components can be effectively chosen by a randomized algorithm. Unfortunately, the paper \cite{EffectiveCounting} does not give exact theoretical bounds that are needed to develop pseudopolynomial algorithms. We present our result in this direction by the following lemma that can be additionally used as a tool to improve the complexity of algorithms, based on Barvinok's works, by a polynomial factor.
\begin{theorem}\label{all_non_zero_th}
Let $\AC$ be a set composed of $m$ non-zero vectors in $\QQ^n$. Then, there exists a randomized algorithm with the expected arithmetic complexity $O(n \cdot m)$, which finds a vector $z \in \ZZ^n$ such that:
\begin{enumerate}
    \item $a^\top z \not= 0$, for any $a \in \AC$;
    \item $\|z\|_{\infty} \leq m$.
\end{enumerate}
The proof is presented in Section \ref{all_non_zero_proof}.
\end{theorem}

All the considered algorithms are called primal-type counting algorithms. Table \ref{tab:comparison_primal} gives a comparison of the considered algorithms.
\begin{table}[h!]
    \centering
    \begin{tabular}{||c|c||}
    \hline
        $\nu \cdot 2^{O(d)} \cdot \bigl(\log_2(\Delta)\bigr)^{d \ln(d)}$ & \cite[Chapter~16]{BarvBook} plus \cite{SVP_exp} \\
        \hline
         $\nu \cdot d^{O(1)} \cdot \Delta^{d}$ & due to \cite{simple_formula_counting} \\
         \hline
         $O\bigl( \nu \cdot d^2 \cdot d^{\log_2(\Delta)} \bigr)$ & due to \cite{CountingFixedM} \\
         \hline
         $O\bigl(\nu^2 \cdot d^4 \cdot \Delta^4 \cdot \log(\Delta) \bigr)$ & {\color{red}this work} \\
         \hline
    \end{tabular}
    \caption{Comparison of different primal-type algorithms}
    \label{tab:comparison_primal}
\end{table}
In our paper, we do not consider dual variants of counting algorithms based on the vector partition function, $Z$-transforms, and complex residue techniques. There are many papers that are based on the dual principle: \cite{beck_residue,beck_partial_fractions,beck_dedekind_sums,knapsack_path_int,IP_complex_int_SP} and \cite{CountingViaContourIntegration,knapsack_lasserre,AlternativeCounting,fixedM_counting_lasserre,knapsack_nesterov}. To best of our knowledge, the pioneers in this direction are the works \cite{brion_simple_poly,brion_residue} due to Brion \& Vergne. For the detailed exposition of dual principles see the monograph \cite{counting_Lasserre_book} due to Lasserre. 

In general, the dual approach is harder to analyse in the sense of computational complexity due to real-valued computation reasons, but in practise this approach shows himself really nice. The paper \cite{AlternativeCounting} due to Lasserre \& Zeron presents a dual-type algorithm with complexity analysis for problem \ref{main_prob} in the standard form. The algorithm's complexity can be estimated by $O(m)^d \cdot \Lambda$, where the parameter $\Lambda$ depends as a polynomial on $m$, $d$, and $\Delta_1$, but exponentially on the input size. 
Let us consider the ILP problem in the canonical form $\max\{c^\top x \colon A x = b,\, x \in \ZZ^n_{\geq 0}\}$, where $A \in \ZZ^{m \times n}$ and $\rank(A) = m$. The goal of the papers \cite{SteinitzILP,OnCanonicalProblems_Grib,DiscConvILP} is to construct efficient FPT-algorithms for this problem with respect to the parameters $m$, $\Delta$, and $\Delta_1$. More precisely, the paper \cite{DiscConvILP} due Jansen \& Rohwedder gives an algorithm with the arithmetic complexity bound $O(m)^m \cdot \Delta_1^{2 m} + T_{LP}$, and the paper \cite{OnCanonicalProblems_Grib} gives an algorithm with the arithmetic complexity bound $O(\log m)^{2 m^2} \cdot m^{m+1} \cdot \Delta^2 \cdot \log(m \Delta) + T_{LP}$. For $m = 1$, this problem is exactly the unbounded knapsack problem $\max\{c^\top x \colon a^\top x = a_0,\, x \in \ZZ^n_{\geq 0}\}$. Consequently, due to \cite{DiscConvILP}, it can be solved in $O(\Delta_1^2)$-time. Note, that the significantly earlier paper \cite{knapsack_nesterov} due to Nesterov gives an $\tilde O(n \Delta_1 + a_0)$-arithmetical complexity algorithm for this problem, which can be more efficient, due to the following reasons. Due to the proximity argument from the work \cite{SteinitzILP} of Eisenbrand \& Weismantel or periodicity reasons from \cite{OnCanonicalProblems_Grib}, it can be assumed that $a_0 \leq \Delta^2_1$; in the natural assumption $n \leq 2 \Delta_1 + 1$, the $\tilde O(n \Delta_1 + a_0)$-complexity algorithm, due to \cite{knapsack_nesterov} can be faster in general, than $O(\Delta_1^2)$-complexity algorithm, due to \cite{DiscConvILP}. Note that the bound $\tilde O(n \Delta_1 + a_0)$ stated in terms of real-valued arithmetic operations. But, it seams that this difficulty in \cite{knapsack_nesterov} can be avoided by a transition to the integer-valued arithmetic.  

So, the next natural question is following: is it possible to count integer points in $\PC = \{x \in \ZZ^n_{\geq 0} \colon A x = b\}$ in FPT-time with respect to the parameters $m$, $\Delta$ or $m$, $\Delta_1$? Corollary \ref{main_corr1} gives a partially positive answer on this question. More precisely, for any fixed $m$, the problem to count $\abs{\PC \cap \ZZ^n}$ can be solved by an FPT-algorithm parameterised by $\Delta$ with the arithmetical complexity bound $O(n)^{2m+4} \cdot \Delta^4 \cdot \log(\Delta)$. A similar parameterized algorithm with respect to $\Delta_1$ can be achieved just by using the Hadamard's bound. For $m=1$, it gives an $O(n^6 \cdot \Delta_1^4 \cdot \log(\Delta_1))$ FPT-algorithm to count the solutions of the unbounded subset-sum problem. For $m=1$, our result is not new, the earlier paper \cite{knapsack_lasserre} due to Lasserre \& Zeron also gives FPT-algorithm for this problem, but the exact complexity bound was not given. Probably, the complexity bound of the paper \cite{knapsack_lasserre} is better, than ours for $m=1$, because our result is more general since it affects general $d$-dimensional $\Delta$-modular simplices that can be defined in the canonical form. Finally, the paper \cite{fixedM_counting_lasserre} due to Lasserre \& Zeron gives a dual-type algorithm with the complexity bound ${\Delta^{O(1)}_{total}} \cdot n^{O(m)}$, where $\Delta_{total}$ is the maximum absolute value of all $k \times k$ sub-determinants of $A$, for $k \in \intint m$. Unfortunately, we can not give exact powers of the exponents, because the computational complexity analysis is not completely finished in \cite{fixedM_counting_lasserre}. The advantage of our result, due to Corollary \ref{main_corr1}, is that our arithmetic complexity bound $O(n)^{2m+4} \cdot \Delta^4 \cdot \log(\Delta)$ depends on the weaker parameter $\Delta$ instead of $\Delta_{total}$.

The ideas of dual-type algorithms and its residue techniques have been significantly modified in the papers \cite{knapsack_path_int,IP_complex_int_SP,CountingViaContourIntegration}. Due to Hirai, Oshiro \& Tanaka \cite{CountingViaContourIntegration}, dual-type algorithms can be significantly more memory-saving, than primal type-algorithms. For example, the existence of an $O(\|b\|^m_{\infty})$-time and $\poly(n,m,\|b\|_{\infty})$-memory algorithm for the multi-dimensional knapsack was shown in \cite{CountingViaContourIntegration}. Here $b \in \ZZ^m_{>0}$ is the r.h.s. vector of the knapsack problem. Note that our approach gives only $O(n^3 \cdot \Delta^3)$-memory usage algorithm for the unbounded multi-dimensional knapsack problem, which can grow exponentially on $m$ as $O(m^{\frac{3}{2} m} \cdot \Delta_1^{3m})$.

\subsection{Auxiliary facts from the polyhedral algebra} \label{algebra_subs}

In this Subsection, we mainly follow to the monographs \cite{BarvBook,BarvPom}. Let $\VC$ be a $d$-dimensional real vector space and $\LC \subset \VC$ be a lattice.

\begin{definition}
Let $\AC \subseteq \VC$ be a set. The \emph{indicator} $[\AC]$ of $\AC$ is the function $[\AC]\colon \VC \to \RR$ defined by
$
[\AC](x) = \begin{cases}
1\text{, if }x \in \AC\\
0\text{, if }x \notin \AC
\end{cases}
$.
The \emph{algebra of polyhedra} $\PS(\VC)$ is the vector space defined as the span of the indicator functions of all the polyhedra $\PC \subset \VC$.
\end{definition}

\begin{definition}
A linear transformation $\TC \colon \PS(\VC) \to \WC$, where $\WC$ is a vector space, is called a \emph{valuation}. We consider only \emph{$\LC$-valuations} or \emph{lattice valuations} that satisfy
$$
\TC([\PC + u]) = \TC([\PC]), \quad \text{for all rational polytopes }\PC \text{ and }u \in \LC,
$$ see \cite[pp. 933--988]{ValuationsAndDissections}, \cite{ValuationsOnConvecBodies}.
\end{definition}




We are mainly interested in two valuations, the first is the counting valuation $\EC([\PC]) = |\PC \cap \ZZ^d|$ and the second valuation $\FC([\PC])$, which will be significantly used in our paper, is defined by the following theorem, proved by J.~Lawrence \cite{Lawrence}, and, independently, 
by A.~Khovanskii and A.~Pukhlikov \cite{Pukhlikov}. We borrowed the formulation from  \cite[Section~13]{BarvBook}:

\begin{theorem}[\cite{Lawrence,Pukhlikov}]\label{rational_gen_th}
Let $\RS(\CC^d)$ be the space of rational functions on $\CC^d$ spanned by the functions of the type 
$$
\frac{\xB^v}{(1 - \xB^{u_1})\dots(1-\xB^{u_d})},
$$
where $v \in \ZZ^d$ and $u_i \in \ZZ^d \setminus \{\BZero\}$, for any $i \in \intint d$. Then there exists a linear transformation (a valuation)
$
\FC \colon \PS(\QQ^d) \to \RS(\CC^d)
$
such that the following properties hold:
\begin{enumerate}
\item Let $\PC \subset \RR^d$ be a non-empty rational polyhedron without lines, and let $\CCal$ be its recession cone. Let $\CCal$ be generated by rays $w_1, \dots, w_n$, for some $w_i \in \ZZ^d \setminus \{\BZero\}$, and let us define 
$$
\WC_{\CCal} = \bigl\{ \xB \in \CC^d \colon |\xB^{w_i}| < 1 \text{,~for~any~} i \in \intint n \bigr\}.
$$
Then, $\WC_{\CCal}$ is a non-empty open set and, for all $\xB \in \WC_{\CCal}$, the series
$$
\fG(\PC; \xB) = \sum\limits_{z \in \PC \cap \ZZ^d} \xB^z
$$ converges absolutely and uniformly on compact subsets of $\WC_{\CCal}$ to the function $f(\PC;\xB) = \FC([\PC]) \in \RS(\CC^d)$.
\item If $P$ contains a line, then $f(\PC;\xB) = 0$.
\end{enumerate}
\end{theorem}
If $\PC$ is a rational polyhedron, then $f(\PC;\xB)$ is called its \emph{short rational generating function}.

\begin{definition}\label{tang_cone_def}
Let $\PC \subset \VC$ be a non-empty polyhedron, and let $v \in \PC$ be a point. We define the \emph{tangent cone} of $\PC$ at $v$ by
$$
\tcone(\PC,v) = \bigl\{v+ y \colon v + \varepsilon y \in \PC, \; \text{ for some } \varepsilon > 0 \bigr\}.
$$
If an $n$-dimensional polyhedron $\PC$ is defined by a system $A x \leq b$, then, for any $v \in \PC$, it holds
\begin{equation*}
\tcone(\PC,v) = \{ x \in \VC \colon A_{\JC(v) *} x  \leq b_{\JC(v)}\},\quad
\text{where $\JC(v) = \{ j \colon A_{j *} v = b_j\}$.}
\end{equation*}

\end{definition}


It is widely known that a slight perturbation in the right-hand sides of a system $A x \leq b$ can transform the polyhedron $\PC(A,b)$ to a simple one. Here, we need an algorithmic version of this fact, presented in the following technical theorem.
\begin{theorem}\label{poly_simplification_th}
Let $A \in \ZZ^{k \times n}$, $\rank(A) = n \leq k$, $b \in \QQ^k$, $\gamma = \max\{\|A\|_{\max},\, \|b\|_{\infty}\}$, and $\PC = \PC(A,b)$ be an $n$-dimensional polyhedron.
Then, for $1/\varepsilon = 1 +  2n \cdot n^{\lceil n/2 \rceil} \cdot \gamma^n $ and the vector $t \in \QQ^k$, with $t_i = \varepsilon^{i-1}$, the polyhedron $\PC' = \PC(A,b+t)$ is simple.
\end{theorem}
\begin{proof}
Let us suppose by the contrary that there exists a vertex $v$ of $\PC'$ and a set of indices $\JC$ such that $A_{\JC} v = (b + t)_{\JC}$, $|\JC| = n + 1$ and $\rank(A_{\JC}) = n$. The last is possible iff $\det(M) = 0$, where $M = \bigl(A_{\JC}\, (b+t)_{\JC} \bigr)$. Note that $M = B + D$, where $B = \bigl(A_{\JC}\, b_{\JC} \bigr)$ and $D = \bigl( \BZero_{(n+1)\times n} \, t_{\JC} \bigr)$. We have,
\begin{multline*}
    \det(M) = \det(B) + \sum\limits_{i = 1}^{n+1} \det(B[i,t_{\JC}]) = \\
    = \det(B) + \sum\limits_{i = 1}^{n+1} \sum\limits_{j = 1}^{n+1} (-1)^{i + j} \cdot (t_{\JC})_j \cdot \det(B_{\JC\setminus\{j\} \IC\setminus\{i\}}),
\end{multline*}
where $\IC = \intint{n+1}$ and $B[i,t_{\JC}]$ is the matrix induced by the substitution of the column $t_{\JC}$ instead of the $i$-th column of $B$. 

Let us assume that $(t_{\JC})_j = \varepsilon^{d_j}$, for $j \in \IC$, where $d_j \in \ZZ$ and $0 \leq d_1 < d_2 < \dots < d_{n+1} \leq k-1$. Consequently, the condition $\det(M) = 0$ is equivalent to the following condition:
\begin{equation}\label{degen_poly}
    \det(B) + \sum\limits_{j=1}^{n+1} \varepsilon^{d_j} \cdot \left( \sum\limits_{i=1}^n (-1)^{i+j} \cdot \det(B_{\JC\setminus\{j\} \IC\setminus\{i\}})\right) = 0.
\end{equation}
Note that the polynomial \eqref{degen_poly} is non-zero. Definitely, since $\rank(A_{\JC}) = n$, we can assume that the first $n$ rows of $A_{\JC}$ are linearly independent. Consequently, there exists a unique vector $y \in \QQ^n_{\not=0}$ such that the last row of $A_{\JC}$ is a linear combination of the first rows with the coefficients vector $y$. Since $\forall \varepsilon\colon \det(M) = 0$, we have $\binom{y}{-1}^\top M = \BZero$ and, consequently, $\binom{y}{-1}^\top (b_{\JC} + t) = \BZero$. But, the last may hold only for a finite number of $\varepsilon$. That is the contradiction.

Using the well known Cauchy's bound, we have that $|\varepsilon^*| \geq \frac{1}{1 + \alpha_{\max}/\beta} = \frac{\beta}{\beta + \alpha_{\max}}$, where $\varepsilon^*$ is any root of \eqref{degen_poly}, $\alpha_{\max}$ is the maximal absolute value of the coefficients, and $\beta$ is the absolute value of the leading coefficient. Finally, $1/|\varepsilon^*| \leq 2 \alpha_{\max} \leq 2n \cdot n^{n/2} \cdot \gamma^n$, which contradicts to the Theorem's condition on $\varepsilon$.
\end{proof}

\section{Proof of Theorem \ref{all_non_zero_th}}\label{all_non_zero_proof}

\subsection{Integer points in intersections of a hypercube with hyperplanes}

\begin{lemma}\label{cube_intnum_lm}
Let $\BC = \bigl\{ x \in \RR^n \colon -r \leq x_i \leq r, \text{ for } i \in \intint n\bigr\}$ be a standard hypercube of volume $(2 r + 1)^n$. Then, for any hyperplane $\HC$, passing throw $\BZero$, we have
$$
|\BC \cap \HC \cap \ZZ^n| \leq (2 r + 1)^{n-1}.
$$
\end{lemma}
\begin{proof}
Assume that $\HC = \{ x \in \RR^n \colon a^\top x = 0\}$, for $a \in \RR^n$ and $a_1 \not= 0$, and denote $\NC = \BC \cap \HC \cap \ZZ^n$. Let $\MC$ be the set, induced by the projection $\proj_0$ of $\NC$ into the hyperplane $x_1 = 0$, that is
$$
\MC = \bigl\{ (0, x_2, \dots, x_n)^\top \in \ZZ^n \colon \exists (x_1, x_2, \dots, x_n)^\top \in \NC \bigr\}.
$$ We claim that the map $\proj_0$ is a bijection between $\NC$ and $\MC$. By the construction, $\proj_0$ is surjective. Let us show that $\proj_0$ is injective. To obtain a contradiction, suppose that there exist different $u,v \in \NC$, such that $\proj_0(u) = \proj_0(v)$. Consequently, $u_1 \not= v_1$ and $u_i = v_i$, for all $i \in \intint[2]{n}$. Then, the following sequence of implications holds:
\begin{equation*}
    a^\top u =0,\; a^\top v = 0 \quad\Longrightarrow\quad a^\top(u - v) = 0 \quad\Longrightarrow\quad a_1 (u_1 - v_1) = 0 \quad\Longrightarrow\quad u_1 = v_1,\\
\end{equation*} which contradicts to $u_1 \not= v_1$ and proves the claim.

Therefore, $|\NC| = |\MC|$. Consider now the set $\NC_0 = \BC \cap \{x \in \ZZ^n \colon x_1 = 0\}$. By the definition of $\MC$: $\MC \subseteq \NC_0$. Consequently,
$$
|\NC| = |\MC| \leq |\NC_0| = (2r+1)^{n-1},
$$ which proves the lemma.
\end{proof}

\subsection{Finishing the proof of Theorem \ref{all_non_zero_th}}
\begin{proof}
Let us fix a parameter $r$ and consider a standard hypercube $\BC = \bigl\{ x \in \RR^n \colon -r \leq x_i \leq r, \text{ for } i \in \intint n\bigr\}$ of volume $(2 r + 1)^n$. For $a \in \AC$, denote $\HC_a = \BC \cap \{x \in \RR^n \colon a^\top x = 0\}$, and let 
$$
\NC = \BC \setminus \bigcup\limits_{a \in \AC} \HC_a.
$$
Due to Lemma \ref{cube_intnum_lm}, we have $|\HC_a \cap \ZZ^n| \leq (2 r + 1)^{n-1}$. Consequently, 
\begin{multline*}
|\NC \cap \ZZ^n| \geq |\BC \cap \ZZ^n| - \sum\limits_{a \in \AC} |\HC_a \cap \ZZ^n| \geq (2 r + 1)^n - m \cdot (2 r + 1)^{n-1} = \\
= (2 r + 1)^{n-1} \cdot (2 r + 1 - m).
\end{multline*}
\begin{equation*}
\text{Therefore,}\quad \frac{|\NC \cap \ZZ^n|}{|\BC \cap \ZZ^n|} \geq \frac{2r +1 - m}{ 2 r + 1} = 1 - \frac{m}{2 r + 1}.
\end{equation*}

Let us assign $r := m$, then the last inequality becomes $\frac{|\NC \cap \ZZ^n|}{|\BC \cap \ZZ^n|} > 1/2$. Now, to find the desired vector $z \in \ZZ^n$, we uniformly sample a point $z \in \BC \cap \ZZ^n$. With a probability at least $1/2$ it will satisfy the first claim of the theorem. The second claim is satisfied just by the construction of $z$. Therefore, the expected number of sampling iterations is $O(1)$. The arithmetic complexity of a single iteration is clearly bounded by $O(n \cdot m)$, which proves the theorem.
\end{proof}

\section{Proof of Theorem \ref{main_th1}}\label{main_th1_sec}

\subsection{A recurrent formula for the generating function of a group polyhedron}

Let $\GC$ be a finite Abelian group and $g_1,\dots,g_n \in \GC$. Let additionally $r_i = \bigl|\langle g_i \rangle\bigr|$ be the order of $g_i$, for $i \in \intint n$, and $r_{\max} = \max_{i} \{r_i\}$. For $g_0 \in \GC$ and $k \in \intint n$, let $\MC(k,g_0)$ be the solutions set of the following system:
\begin{equation}\label{f_k_system}
    \begin{cases}
    \sum\limits_{i = 1}^k x_i g_i = g_0\\
    x \in \ZZ_{\geq 0}^k.
    \end{cases}
\end{equation}
Consider the formal power series 
$
\fG_k(g_0;\xB) = \sum\limits_{z \in \MC(k,g_0) \cap \ZZ^k} \xB^z.
$
For $k = 1$, we clearly have
$$
\fG_1(g_0;\xB) = \frac{x_1^s}{1 - x_1^{r_1}},\quad\text{where $s = \min\{x_1 \in \ZZ_{\geq 0} \colon x_1 g_1 = g_0\}$.}
$$ If such $s$ does not exist, we put $\fG_1(g_0;\xB) = 0$.
Note that, for any value of $x_k \in \ZZ_{\geq 0}$, the system \eqref{f_k_system} can be rewritten as
\begin{equation*}
    \begin{cases}
    \sum\limits_{i = 1}^{k-1} x_i g_i = g_0 - x_k g_k\\
    x \in \ZZ_{\geq 0}^{k-1}.
    \end{cases}
\end{equation*}
Hence, for $k \geq 1$, we have
\begin{multline}\label{f_k_recurrence}
    \fG_k(g_0;\xB) = \\
    = \frac{ \fG_{k-1}(g_0;\xB) + x_{k} \cdot \fG_{k-1}(g_0 - g_k;\xB) + \dots + x_{k}^{r_k - 1} \cdot \fG_{k-1}(g_0 - g_k \cdot (r_k - 1);\xB)} {1 - x_k^{r_k}} = \\
    = \frac{1}{1 - x_{k}^{r_k}} \cdot \sum_{i = 0}^{r_k - 1} x_k^i \cdot \fG_{k-1}(g_0 - i \cdot g_k;\xB).
\end{multline}

\begin{equation}\label{f_k_conv}
\text{Consequently,} \quad \fG_k(g_0;\xB) = \frac{\sum\limits_{i_1 = 0}^{r_1-1}\dots\sum\limits_{i_k = 0}^{r_k-1} \epsilon_{i_1,\dots,i_k} x_1^{i_1} \dots x_k^{i_k}}{(1 - x_1^{r_1})(1 - x_2^{r_2})\dots(1 - x_k^{r_k})},    
\end{equation}
where the numerator is a polynomial with coefficients $\epsilon_{i_1,\dots,i_k} \in \{0,1\}$ and degree at most $(r_1 - 1) \dots (r_k - 1)$. Additionally, the formal power series $\fG_k(g_0;\xB)$ converges absolutely to the given rational function if $|x_i^{r_i}| < 1$, for each $i \in \intint k$.  

\subsection{Simple $\Delta$-modular polyhedral cone and its generating function}

Let $A \in \ZZ^{n \times n}$, $b \in \ZZ^n$, $\Delta = |\det(A)| > 0$, $A^* = \Delta \cdot A^{-1}$, $\PC = \PC(A,b)$, and let us consider the formal power series
$$
\fG(\PC;\xB) = \sum\limits_{z \in \PC \cap \ZZ^{n}} \xB^z.
$$
Let $A = P^{-1} S Q^{-1}$ and $\sigma = S_{n n} = \Delta / \Delta_{\gcd}(A,n-1)$, where $S \in \ZZ^{n \times n}$ is the SNF of $A$ and $P,Q \in \ZZ^{n \times n}$ are unimodular matrices. After the unimodular map $x = Q x'$ and introducing slack variables $y$, the system $A x \leq b$ becomes 
$$
\begin{cases}
S x + P y = P b\\
x \in \ZZ^{n}\\
y \in \ZZ^{n}_{\geq 0}.
\end{cases}
$$
Since $P$ is unimodular, the last system is equivalent to the system
\begin{equation}\label{simple_system_group}
\begin{cases}
P y = P b \pmod{S \cdot \ZZ^n}\\
y \in \ZZ^{n}_{\geq 0}.
\end{cases}
\end{equation}
Note that points of $\PC \cap \ZZ^n$ and the system \eqref{simple_system_group} are connected by the bijective map $x = A^{-1}(b - y)$. The system \eqref{simple_system_group} can be interpreted as a group system \eqref{f_k_system}, where $\GC = \ZZ^{n} \bmod\, S$ with an addition modulo $S$, $k = n$, $g_0 = P b \bmod S$ and $g_i = P_{* i} \bmod S$, for $i \in \intint{n}$. Clearly, $\GC$ is isomorphic to $\ZZ^n/S\cdot\ZZ^n$, $|\GC| = |\det(S)| = \Delta$ and $r_{\max} \leq \sigma$.

Following the previous subsection, for $k \in \intint{n}$ and $g_0 \in \GC$, let $\MC_k(g_0)$ be the solutions set of the system $$
\begin{cases}
\sum\limits_{i = 1}^k y_i g_i = g_0\\
y \in \ZZ_{\geq 0}^k,
\end{cases} \quad\text{and}\quad
\fG_k(g_0; \xB) = \sum\limits_{y \in \MC_k(g_0)} \xB^{-\sum\limits_{i=1}^k h_{i} y_i},
$$ where $h_i$ is the $i$-th column of the matrix $A^*$.
Note that
\begin{multline}\label{group_connection_formula}
    \fG(\PC;\xB) = \sum\limits_{z \in \PC \cap \ZZ^{n}} \xB^z
    = \sum\limits_{y \in \MC_n(P b \bmod S)} \xB^{A^{-1}(b-y)} = \\
    = \xB^{A^{-1} b} \cdot \sum\limits_{y \in \MC_n(P b \bmod S)} \xB^{-\frac{1}{\Delta} A^* y} =  \xB^{A^{-1} b} \cdot \fG_n\bigl(P b \bmod S;\xB^{\frac{1}{\Delta}}\bigr).
\end{multline}
Next, we will use the formulas \eqref{f_k_recurrence} and \eqref{f_k_conv} after the substitution $x_i \to \xB^{- h_i}$, for $i \in \intint n$. For $k = 1$, we have
\begin{equation}\label{ff_k_initial}
\fG_1(g_0; \xB) = \frac{\xB^{- s h_1}}{1 - \xB^{-r_1 h_1}},\quad\text{where $s = \min\{y_1 \in \ZZ_{\geq 0} \colon y_1 g_1 = g_0 \}$.}
\end{equation}
For $k \geq 2$, we have
\begin{gather}
\fG_k(g_0;\xB) = \frac{1}{1 - \xB^{-r_k h_k}} \cdot \sum\limits\limits_{i = 0}^{r_k-1} \xB^{- i h_k} \cdot \fG_{k-1}(g_0 - i \cdot g_k; \xB) \text{ and}\label{ff_k_recur}\\    
\fG_k(g_0;\xB) = \frac{\sum\limits_{i_1 = 0}^{r_1-1}\dots\sum\limits_{i_k = 0}^{r_k-1} \epsilon_{i_1,\dots,i_k} \xB^{-(i_1 h_1 + \dots + i_k h_k)}}{(1 - \xB^{-r_1 h_1})(1 - \xB^{-r_2 h_2}) \dots (1 - \xB^{-r_k h_k})}\label{ff_k_conv},
\end{gather} where the numerator is a Laurent polynomial with coefficients $\epsilon_{i_1,\dots,i_k} \in \{0,1\}$. Clearly, the power series $\fG_k(g_0;\xB)$ converges absolutely to the given function if $|\xB^{-r_i h_i}| < 1$, for each $i \in \intint[1]{k}$. 
Due to the formulae \eqref{ff_k_conv} and \eqref{group_connection_formula}, we have
\begin{equation}\label{ff_PC_coeff}
    \fG(\PC;\xB) = \frac{\sum\limits_{i_1 = 0}^{r_1-1}\dots\sum\limits_{i_n = 0}^{r_n-1} \epsilon_{i_1,\dots,i_n} \xB^{\frac{1}{\Delta}A^*(b - (i_1,\dots,i_n)^\top)} }{\bigl(1 - \xB^{-\frac{r_1}{\Delta} h_1}\bigr)\bigl(1 - \xB^{-\frac{r_2}{\Delta} h_2}\bigr) \dots \bigl(1 - \xB^{-\frac{r_n}{\Delta} h_n}\bigr)}.
\end{equation}

Note that $\frac{r_i}{\Delta} h_i$ is an integer vector, for any $i \in \intint n$, and $\frac{1}{\Delta}A^*(b - (i_1,\dots,i_n)^\top)$ is an integer vector, for any $(i_1,\dots,i_n)$, such that $\epsilon_{i_1,\dots,i_n} \not= 0$. Indeed, by definition of $r_i$, we have $r_i P_{* i} \equiv \BZero \pmod{S \cdot \ZZ^n}$, so $\frac{r_i}{\Delta} h_i = (r_i A^{-1})_{* i} = (Q S^{-1} P r_i)_{* i}$, which is an integer vector. Vectors $(i_1,\dots,i_n)^\top$ correspond to solutions $y$ of the system \eqref{simple_system_group}, and $\frac{1}{\Delta}A^*(b - (i_1,\dots,i_n)^\top) = A^{-1}(b - y)$ is an integer vector. Additionally, note that the vectors $-\frac{r_i}{\Delta} h_i$ represent extreme rays of the recession cone of $\PC$.

Let $c \in \ZZ^{n}$ be chosen, such that $c^\top h_i \not= 0$, for any $i$. Consider the exponential sum
$$
\hat \fG_k(g_0;\tau) = \sum\limits_{y \in \MC_k(g_0)} e^{- \tau \cdot \langle c, \sum_{i=1}^k h_i y_i \rangle}
$$ that is induced from $\fG_k(g_0;\xB)$ by the substitution $x_i = e^{\tau \cdot c_i}$.
The formulae \eqref{ff_k_initial}, \eqref{ff_k_recur}, and \eqref{ff_k_conv} become
\begin{gather}
    \hat \fG_1(g_0; \tau) = \frac{e^{- \langle c, s  h_1 \rangle \cdot \tau}}{1 - e^{- \langle c, r_1  h_1 \rangle \cdot \tau}},\label{ff_k_tau_initial}\\
    \hat \fG_k(g_0;\tau) = \frac{1}{1 - e^{-\langle c, r_k  h_k \rangle \cdot \tau}} \cdot \sum\limits_{i = 0}^{r_k-1} e^{- \langle c, i h_k \rangle \cdot \tau} \cdot \hat \fG_{k-1}(g_0 - i \cdot g_k; \tau),\label{ff_k_tau_recur}\\
    \hat \fG_k(g_0;\tau) = \frac{\sum\limits_{i_1 = 0}^{r_1-1}\dots\sum\limits_{i_k = 0}^{r_k-1} \epsilon_{i_1,\dots,i_k} e^{-\langle c, i_1 h_1 + \dots + i_k h_k \rangle \cdot \tau} }{\bigl(1 - e^{-\langle c, r_1 h_1 \rangle \cdot \tau}\bigr)\bigl(1 - e^{-\langle c, r_2 h_2 \rangle \cdot \tau}\bigr) \dots \bigl(1 - e^{- \langle c, r_k h_k \rangle \cdot \tau}\bigr)}.\label{ff_k_tau_conv}
\end{gather}
Let $\chi = \max\limits_{i \in \intint n}\bigl\{|\langle c, h_i \rangle|\bigr\}$. Since $\langle c,h_i \rangle \in \ZZ_{\not=0}$, for each $i$, the number of terms $e^{-\langle c, \cdot \rangle \cdot \tau}$ is bounded by $1 + 2 \cdot k \cdot r_{\max}\cdot \chi \leq 1 + 2 \cdot k \cdot \sigma\cdot \chi$. So, after combining similar terms, the numerator's length becomes $O(k \cdot \sigma\cdot \chi)$.
In other words, there exist coefficients $\epsilon_i \in \ZZ_{\geq 0}$, such that 
\begin{equation}\label{ff_k_tau_coef}
\hat \fG_k(g_0;\tau) = \frac{\sum\limits_{i = - k \cdot \sigma \cdot \chi}^{k \cdot \sigma \cdot \chi} \epsilon_i \cdot e^{- i \cdot \tau}}{\bigl(1 - e^{-\langle c, r_1 \cdot h_1 \rangle \tau}\bigr)\bigl(1 - e^{-\langle c, r_2 h_2 \rangle \cdot \tau}\bigr) \dots \bigl(1 - e^{- \langle c, r_k h_k \rangle \cdot \tau}\bigr)}.    
\end{equation}


Let us discuss the group-operations complexity issues to find the representation \eqref{ff_k_tau_coef} of $\hat \fG_k(g_0;\tau)$, for any $k \in \intint n$ and $g_0 \in \GC$.
Clearly, to find the desired representation of $\hat \fG_1(g_0;\tau)$, for all $g_0 \in \GC$, we need $r_1 \cdot \Delta$ group operations.
Fix $g_0 \in \GC$ and $k \in \intint n$. To find $\hat \fG_k(g_0;\tau)$, for $k \geq 2$, we can use the formula \eqref{ff_k_tau_recur}. Each numerator of the term $e^{- \langle c, i h_k \rangle \cdot \tau} \cdot \hat \fG_{k-1}(g_0-i g_k;\tau)$ contains at most $1 + 2 \cdot (k-1) \cdot \sigma \cdot \chi$ non-zero terms of the type $\epsilon \cdot e^{- \langle c, \cdot \rangle \cdot \tau}$. Hence, the summation can be done with $O(k \cdot \sigma^2 \cdot \chi)$ group operations. Consequently, the total group-operations complexity can be expressed by the formula
$$
O(\Delta \cdot n^2 \cdot \sigma^2 \cdot \chi).
$$
Since the diagonal matrix $S$ can have at most $\log_2(\Delta)$ terms that are not equal to $1$, the arithmetic complexity of a single group operation is $O(\log(\Delta))$. Hence, the total arithmetic complexity is
$$
O\bigl(\Delta \cdot \log(\Delta) \cdot n^2 \cdot \sigma^2 \cdot \chi\bigr).
$$
Finally, let us show how to find the exponential form 
$$
\hat \fG(\PC;\tau) = \sum\limits_{z \in \PC \cap \ZZ^n} e^{\langle c,z \rangle \cdot \tau}
$$ of the power series $\fG(\PC;\xB)$ induced by the map $x_i = e^{c_i \cdot \tau}$.
Due to the formula \eqref{group_connection_formula}, we have
$$
\hat \fG(\PC;\tau) = e^{\langle c, A^{-1} b \rangle \cdot \tau} \cdot \hat \fG_n\bigl(P b \bmod S; \frac{\tau}{\Delta}\bigr).
$$
Due to the last formula and the formulae \eqref{ff_PC_coeff} and \eqref{ff_k_tau_coef}, we have
\begin{equation*}
    \hat \fG(\PC;\tau) = \frac{\sum\limits_{i = - n \cdot \sigma \cdot \chi}^{n \cdot \sigma\cdot \chi} \epsilon_i \cdot e^{\frac{1}{\Delta}(\langle c, A^* b \rangle- i) \cdot \tau}}{\bigl(1 - e^{-\langle c, \frac{r_1}{\Delta} \cdot h_1 \rangle \cdot \tau}\bigr)\bigl(1 - e^{-\langle c, \frac{r_2}{\Delta} \cdot h_2 \rangle \cdot \tau}\bigr) \dots \bigl(1 - e^{- \langle c, \frac{r_n}{\Delta} \cdot h_n \rangle \cdot \tau}\bigr)}.
\end{equation*}  
Again, due to \eqref{ff_PC_coeff}, we have $\langle c, \frac{r_i}{\Delta} h_i \rangle \in \ZZ_{\not=0}$,  for any $i \in \intint n$, and $\frac{1}{\Delta}(\langle c, A^* b \rangle- i) \in \ZZ$, for any $i$, such that $\epsilon_i > 0$. 
Therefore, we have proven the following:
\begin{theorem}\label{eff_representation_th}
Let $A \in \ZZ^{n \times n}$, $b \in \ZZ^n$, $\Delta = |\det(A)| > 0$, and $\PC = \PC(A,b)$. Let, additionally, $\sigma = S_{n n}$, where $S$ is the SNF of $A$, and $\chi = \max\limits_{i \in \intint n}\bigl\{|\langle c, h_i \rangle|\bigr\}$, where $h_i$ is the $i$-th column of $\Delta \cdot A^{-1}$.
Then, the formal exponential series $\hat \fG(\PC;\tau)$ can be represented as
$$
\hat \fG(\PC;\tau) = \frac{\sum\limits_{i = -n \cdot \sigma\cdot \chi}^{n \cdot \sigma\cdot \chi} \epsilon_i \cdot e^{\alpha_i \cdot \tau}}{\bigl(1 - e^{-\beta_1 \cdot \tau}\bigr)\bigl(1 - e^{-\beta_2 \cdot \tau}\bigr) \dots \bigl(1 - e^{-\beta_n \cdot \tau}\bigr)}, 
$$ where $\epsilon_i \in \ZZ_{\geq 0}$, $\beta_i \in \ZZ_{>0}$, and $\alpha_i \in \ZZ$.
This representation can be found with an algorithm having the arithmetic complexity bound
$$
O\bigl(T_{\SNF}(n) + \Delta \cdot \log(\Delta) \cdot n^2 \cdot \sigma^2 \cdot \chi\bigr),
$$ where $T_{SNF}(n)$ is the arithmetical complexity of computing the SNF for $n \times n$ integer matrices.
\end{theorem}

\subsection{Handling the general case}

Following Remark \ref{forms_reduction_rm}, we will only work with polytopes $\PC$ defined with systems in the canonical form. Denote $\gamma = \max\{\|A\|_{\max}, \|b\|_{\infty}\}$, $\beta = \min\limits_{i \in \intint{n+m}}\{\lceil b_i \rceil - b_i \colon b_i \notin \ZZ \}$, and $\varepsilon = \min\{\beta/2,\, (1+2n \cdot n^{\lceil n/2 \rceil} \cdot \gamma)^{-1}\}$. If all $b_i$ are integer, we put $\beta = +\infty$, so the formula for $\varepsilon$ remains correct. Then, by Theorem \ref{poly_simplification_th}, the polytope $\PC' = \PC(A,b+t)$ is simple, where the vector $t$ is chosen, such that $t_i = \varepsilon^{i-1}$, for $i \in \intint{n+m}$. By the construction, $\PC \cap \ZZ^n = \PC' \cap \ZZ^n$. From this moment, we assume that $\PC$ is a simple polytope.

Using Definition \ref{tang_cone_def} for tangent cones, the Brion's Theorem \cite{Brion} (see also \cite[Chapter~6]{BarvBook}) gives:
\begin{multline*}
[\PC] = \sum\limits_{v \in \vertex(\PC)} \bigl[\tcone(\PC,v)\bigr] = \\ 
= \sum\limits_{v \in \vertex(\PC)} \bigl[\PC(A_{\JC(v)},b_{\JC(v)})\bigr] \lmod,
\end{multline*}
where $\JC(v) = \{j \colon A_{j *} v = b_j\}$. Due to the seminal work \cite{AvisFukuda} of Avis \& Fukuda, all vertices of the simple polyhedron $\PC$ can be enumerated with $O\bigl( (m+n) \cdot n \cdot |\vertex(\PC)| \bigr)$ arithmetic operations. Due to Lee, Paat, Stallknecht \& Xu \cite{ModularDiffColumns}, it can be assumed that $n+m = O(n^2 \cdot \Delta^2)$. Therefore, the vertices of $\PC$ can be enumerated with $O(\nu \cdot n^3 \cdot \Delta^2)$ operations, which is negligible with respect to the desired complexity bound.
Denote $f(\PC; \xB) = \FC([\PC]) \in \RS(\QQ^n)$, for any rational polyhedron $\PC$, where $\FC$ is the evaluation considered in Theorem \ref{rational_gen_th}.
Note that $f(\PC(B, u); \xB) = f(\PC(B, \lfloor u \rfloor); \xB)$, for any $B \in \ZZ^{n \times n}$ and $u \in \QQ^n$. So, due to Theorem \ref{rational_gen_th}, we can write
\begin{equation*}
    f(\PC;\xB) = \sum\limits_{v \in \vertex(\PC)} f\Bigl(\PC\bigl(A_{\JC(v)}, \lfloor b_{\JC(v)} \rfloor \bigr); \xB\Bigr).
\end{equation*}
Due to results of the previous subsection, each term $f\Bigl(\PC\bigl(A_{\JC(v)}, \lfloor b_{\JC(v)} \rfloor \bigr); \xB\Bigr)$ has the form \eqref{ff_PC_coeff}.
To find the value of $|\PC \cap \ZZ^n| = \lim\limits_{\xB \to \BUnit} f(\PC;\xB)$, we follow Chapters~13 and 14 of \cite{BarvBook}. Define the set $\EC$ of \emph{edge directions} by the following way
$$
h \in \EC \;\Longleftrightarrow\; h \text{ is a column of } -A^{*}_{\JC(v)} \text{ for some $v \in \vertex(\PC)$},
$$ where $B^* = |\det(B)| \cdot B^{-1}$, for arbitrary invertable $B$. Assume that a vector $c \in \ZZ^n$ is chosen, such that $c^\top h \not= 0$, for each $h \in \EC$. Denote additionally $\chi = \max\limits_{h \in \EC} \bigl\{|c^\top h|\bigr\}$. Substituting $x_i = e^{c_i \cdot \tau}$, let us consider the exponential function
\begin{equation*}
    \hat f(\PC;\tau) = \sum\limits_{v \in \vertex(\PC)} \hat f\Bigl(\PC\bigl(A_{\JC(v)}, \lfloor b_{\JC(v)} \rfloor \bigr); \tau \Bigr).
\end{equation*}
Due to \cite[Chapter~14]{BarvBook}, the value $|\PC \cap \ZZ^n|$ is a constant term in the Tailor series of the function $\hat f(\PC;\tau)$, so we just need to compute it.
Let us fix some term $\hat f\bigl(\PC(B, u); \tau\bigr)$ of the previous formula. Due to Theorem \ref{eff_representation_th}, it can be represented as 
$$
\hat f\bigl(\PC(B, u); \tau\bigr) = \frac{\sum\limits_{i = - n \cdot \sigma\cdot \chi}^{n \cdot \sigma\cdot \chi} \epsilon_i \cdot e^{\alpha_i \cdot \tau}}{\bigl(1 - e^{-\beta_1 \cdot \tau}\bigr)\bigl(1 - e^{-\beta_2 \cdot \tau}\bigr) \dots \bigl(1 - e^{-\beta_n \cdot \tau}\bigr)}, 
$$ where $\epsilon_i \in \ZZ_{\geq 0}$, $\beta_i \in \ZZ_{>0}$, and $\alpha_i \in \ZZ$.

Again, due to \cite[Chapter~14]{BarvBook}, we can see that the constant term in Tailor series for $\hat f\bigl(\PC(B, u); \tau\bigr)$ is exactly
\begin{equation}\label{constant_Tailor}
\sum\limits_{i = - n \cdot \sigma\cdot \chi}^{n \cdot \sigma \cdot \chi} \frac{\epsilon_i}{\beta_1 \dots \beta_n} \sum\limits_{j = 0}^n \frac{\alpha_i^j}{j!} \cdot \toddp_{n-j}(\beta_1, \dots, \beta_n),
\end{equation}
where $\toddp_j(\beta_1,\dots,\beta_n)$ is a homogeneous polynomial of degree $j$, called the \emph{$j$-th Todd polynomial} on $\beta_1,\dots,\beta_n$. 
Due to \cite[Theorem~7.2.8, p.~137]{AlgebracILP}, the values of $\toddp_{j}(\beta_1,\dots,\beta_n)$, for $j \in \intint n$, can be computed with an algorithm that is polynomial in $n$, and the bit-encoding length of $\beta_1,\dots,\beta_n$. Moreover, it follows from the theorem's proof that the arithmetical complexity can be bounded by $O(n^3)$.   
Since $\sigma \leq \Delta$, due to Theorem \ref{eff_representation_th}, the total arithmetic complexity to find the value of \eqref{constant_Tailor} can be bounded by
$$
O\bigl(n^3 + T_{SNF}(n) + \Delta^3 \cdot \log(\Delta) \cdot n^2 \cdot \chi \bigr).
$$
Due to \cite{SNFOptAlg}, $T_{SNF}(n) = O(n^3)$. Assuming that $O(n^2 \cdot \chi)$ dominates $O(n^3)$, the last bound can be rewritten to
$
O\bigl(\Delta^3 \cdot \log(\Delta) \cdot n^2 \cdot \chi \bigr).
$
The constant term in Tailor series for the complete function $\hat f(\PC; \tau)$ can be found just by summation. It gives the arithmetic complexity bound
\begin{equation}\label{algo_complexity_chi}
O\bigl(\nu \cdot n^2 \cdot \Delta^3 \cdot \log(\Delta) \cdot \chi \bigr).    
\end{equation}

\subsection{How to choose a hyperplane that avoids all the edge directions?}

In the previous subsection, we made the assumption that the vector $c \in \ZZ^n$ is chosen such that $c^\top h \not= 0$, for any $h \in \EC$. Let us present an algorithm that generates a vector $c$ with a respectively small value of the parameter $\chi = \max\limits_{h \in \EC} \bigl\{|c^\top h|\bigr\}$. The main idea is concentrated in Theorem \ref{all_non_zero_th}. 

Since the polytope $\PC$ is assumed to be simple, each vertex $v \in \vertex(\PC)$ corresponds to exactly $n$ edge directions. Consequently, $2 \cdot |\EC| = \nu \cdot n$. Choose some basis sub-matrix $B$ of $A$. Note that $B h \not= \BZero$ and $(B h)_i \in \intint[-\Delta]{\Delta}$, for any $h \in \EC$ and $i \in \intint n$. Next, we use Theorem \ref{all_non_zero_th} to the set $B \cdot \EC$, which produces a vector $z$, such that 
\begin{enumerate}
    \item $z^\top B h \not= 0$, for each $h \in \EC$;
    \item $\|z\|_{\infty} \leq \nu \cdot n$.
\end{enumerate}
Now, we assign $c := B^\top z$. By the construction, we have $c^\top h \not= 0$ and $|c^\top h| = |z^\top B h| \leq n^2 \cdot \nu \cdot \Delta$, for each $h \in \EC$. Therefore, $\chi \leq n^2 \cdot \nu \cdot \Delta$.

\subsection{What is the total algorithm arithmetic complexity?}
Finally, let as estimate the total algorithm complexity. Combining the formula \eqref{algo_complexity_chi} with our bound for $\chi$, it gives
$$
O\bigl(\nu^2 \cdot n^4 \cdot \Delta ^4 \cdot \log(\Delta)\bigr),
$$ which proves Theorem \ref{main_th1}.

\subsection{Proof of Corollary \ref{main_corr1}}\label{main_corr1_subs}
The presented complexity bounds follow from the different ways to estimate the value $\nu$. The first bound trivially follows from the inequalities $\nu \leq \binom{n+m}{n} = \binom{n+m}{m} \lesssim \frac{e^m \cdot (n+m)^m}{m^m} = O\bigl(\frac{n}{m} + 1\bigr)^m$. 
To obtain the second bound, we refer to the seminal result, due to McMullen~\cite{MaxFacesTh}. Together with the formula from \cite[Section~4.7]{Grunbaum} for the number of facets of a cyclic polytope, it follows that the maximal number of vertices in an $n$-dimensional polyhedron with $k$ facets is bounded by $$\xi(n,k) = \begin{cases}
    \frac{k}{k-s} \binom{k-s}{s},\text{ for }n = 2s\\
    2\binom{k-s-1}{s},\text{ for }n = 2s+1\\
    \end{cases} = O\left(\frac{k}{n}\right)^{n/2}.$$
Clearly, $\nu \leq \xi(n,n+m)$, and $\nu = O\bigl( \frac{n+m}{n} \bigr)^{\frac{n}{2}}$. So, the second bound holds. Due to Lee, Paat, Stallknecht \& Xu \cite{ModularDiffColumns}, we can assume that $n+m = O(n^2 \cdot \Delta^2)$. Substituting the last formula to the second bound, we obtain $\nu = O(n^{\frac{n}{2}} \cdot \Delta^n)$, and the third bound holds. 

Finally, let us show how to handle the case, when $\PC$ is an unbounded $n$-dimensional polyhedron. Clearly, we need to distinguish between two possibilities: $|\PC \cap \ZZ^n| = 0$ and $|\PC \cap \ZZ^n| = \infty$.  Let us choose any vertex $v$ of $\PC$ and consider a set of indices $\JC$, such that $|\JC| = n$, $A_{\JC} v = b_{\JC}$ and $\rank(A_{\JC}) = n$. For the first and second bounds, we add a new inequality $c^\top x \leq c_0$ to the system $A x \leq b$, where $c^\top = \sum_{i=1}^n (A_{\JC})_{i *}$ and $c_0 = c^\top v + \|c\|_1 \cdot n \Delta + 1$. Let $A' x \leq b'$ be the new system. Due to the seminal ILP sensitivity bound of Cook, Gerards, Schrijver \& Tardos \cite{Sensitivity_Tardos}, $|\PC \cap \ZZ^n| = 0$ iff $|\PC(A',b') \cap \ZZ^n| = 0$. Since $\PC(A',b')$ is a polytope and $\Delta(A') \leq n \Delta$, we just need to add an additional multiplicative factor of $O(\frac{d}{m} +1)^2 \cdot n^4$ to the first bound and $O(n^4)$ to the second bound. To deal with third bound, we just need to add additional inequalities $A_{\JC} x \geq b_{\JC} - \|A_{\JC}\|_{\max} \cdot n^2 \Delta \cdot \BUnit$ to the system $A x \leq b$. The polyhedron becomes bounded and the sub-determinants stay unchanged, and we follow the original scenario.

\bibliographystyle{amsplain}
\bibliography{grib_biblio}

\end{document}